%% file: main.tex
\documentclass[10pt,twocolumn,letterpaper]{article}

\usepackage[pagenumbers]{wacv}      

\input{preamble}

\definecolor{wacvblue}{rgb}{0.21,0.49,0.74}
\usepackage[pagebackref,breaklinks,colorlinks,allcolors=wacvblue]{hyperref}

\def\confName{WACV}
\def\confYear{2027}

\title{Measured-Subspace Consistency: A Plug-and-Play Operator for Diffusion Posterior Sampling in Accelerated MRI Reconstruction}

\author{Junhyeok Lee$^{1}$\qquad Kyu Sung Choi$^{2,3}$\\
$^{1}$Interdisciplinary Program in Cancer Biology, Seoul National University College of Medicine\\
$^{2}$Department of Radiology, Seoul National University College of Medicine, Seoul National University Hospital\\
$^{3}$Healthcare AI Research Institute, Seoul National University Hospital}

\begin{document}
\maketitle
\input{sec/0_abstract}
\input{sec/1_intro}
\input{sec/2_related}
\input{sec/3_method}

\input{sec/4_experiments}
\input{sec/5_conclusion}
{
    \small
    \bibliographystyle{ieeenat_fullname}
    \bibliography{main}
}

\end{document}

%% file: preamble.tex
\usepackage{multirow}
\usepackage{hhline}
\usepackage{rotating}
\usepackage{pifont}
\usepackage{amsthm}
\usepackage{algorithm}
\usepackage{algpseudocode}
\newcommand{\cmark}{\ding{51}}

\newtheorem{proposition}{Proposition}

\newcommand{\Pmsc}{\mathcal{P}_{\mathrm{MSC}}}
\newcommand{\Pideal}{\mathcal{P}^{*}}
\newcommand{\vmeas}{\mathcal{D}_{\mathrm{meas}}}
\newcommand{\vunmeas}{\mathcal{D}_{\mathrm{unmeas}}}

\newcommand{\x}{\mathbf{x}}
\newcommand{\y}{\mathbf{y}}

%% file: sec/0_abstract.tex
\begin{abstract}
Diffusion posterior samplers for accelerated MRI can reconstruct accurately yet still disagree on the acquired k-space across samples, placing posterior variability on coefficients the scanner has already measured. We identify this measured-subspace leakage as a physical-admissibility failure. Under a hard-constraint model it violates the measurement constraint and inflates the reported uncertainty with disagreement about coefficients the scanner has already determined. To quantify this leakage, we introduce complementary measured- and unmeasured-subspace k-space dispersion metrics (MSD/USD). We then present Measured-Subspace Consistency (MSC), a training-free terminal correction that wraps any compatible image-space posterior sampler with a standard multi-coil consistency lock. The ideal lock follows classical range/null-space data consistency. Our contribution is to repurpose it as a black-box posterior audit and correction rather than a new reconstructor or learned sampler. Theoretically, we prove that the ideal transform confines pairwise sample differences to the MRI null space and bound the residual cross-subspace coupling left by practical sensitivity-weighted implementations. Across six base samplers and two MRI anatomies, including out-of-distribution transfer where a knee prior reconstructs brain, MSC substantially reduces measured-subspace dispersion for Soft samplers (a median $16.5\times$ reduction for DPS across five brain contrasts, up to ${\sim}29\times$), while preserving unmeasured-subspace diversity and acting as a near-identity map for Consistent ones. Furthermore, MSC maintains or modestly improves PSNR/SSIM, with no retraining, retuning, or significant computational overhead. 
\end{abstract}

\begin{figure}[t]
\centering
\includegraphics[width=0.98\linewidth]{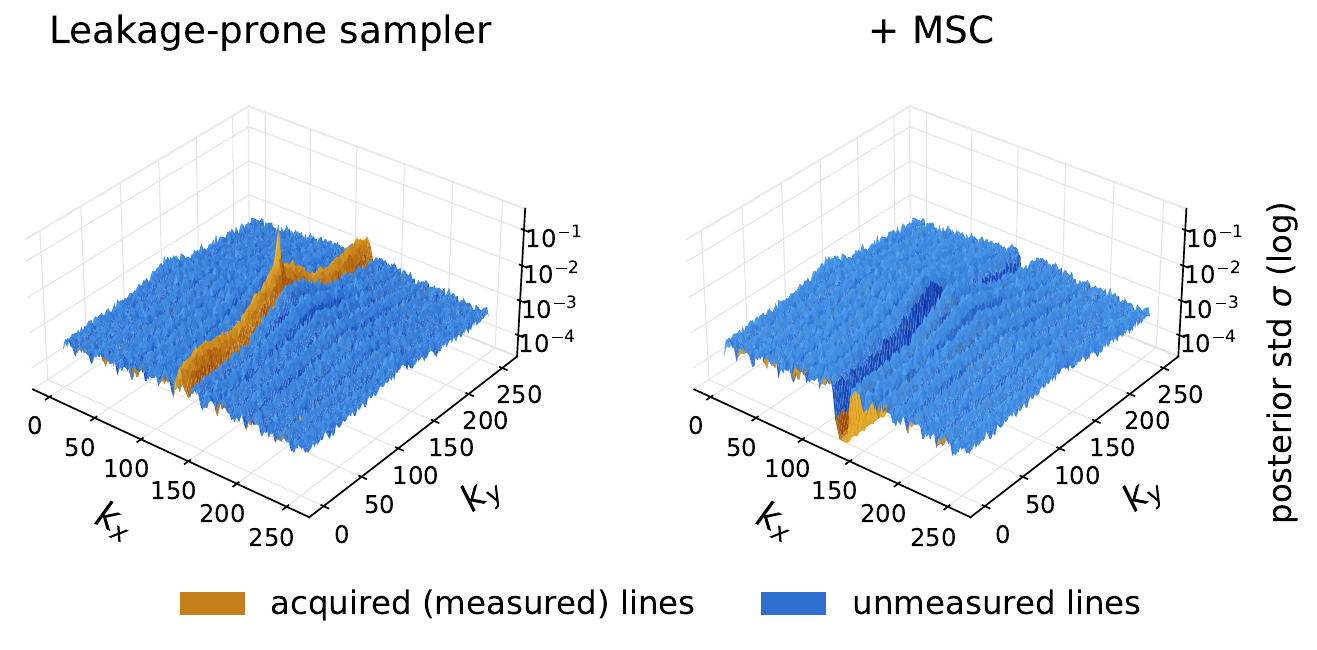}
\caption{\textbf{Measured-subspace leakage and MSC.} Per-$(k_x,k_y)$ posterior standard deviation across samples in multi-coil k-space. Orange lines are acquired, blue lines are unmeasured, and the ACS peak is clipped for visibility. MSC locks acquired lines, reducing inadmissible measured-line variability while retaining unmeasured diversity up to the practical coil-combination residual.}
\label{fig:teaser}
\end{figure}

%% file: sec/1_intro.tex
\section{Introduction}
\label{sec:intro}
Accelerated MRI reconstruction has evolved from compressed sensing with hand-crafted sparsity priors~\cite{lustig2007sparse} to physics-driven deep networks that unroll the reconstruction with learned regularizers~\cite{hammernik2018varnet,aggarwal2019modl,knoll2020deeplearning}. Most recently, score-based diffusion models~\cite{song2019ncsn,ho2020ddpm,song2021scoresde} have emerged as powerful priors for this task~\cite{jalal2021robust,chung2022scoremri,song2021solving}. Used as posterior samplers, they produce not a single reconstruction but a distribution of plausible images, enabling valuable uncertainty quantification. However, high image-space reconstruction quality (e.g., PSNR, SSIM) does not guarantee a physically valid posterior. In practice, a sampler can produce visually convincing images while stochastically varying the k-space coefficients that the scanner has \emph{already} measured. Since the acquired data are fixed observations, sample-to-sample variability on those coefficients is physically inadmissible in the hard-constraint (noise-negligible) regime we take as our reference. Under a noisy acquisition a calibrated posterior may instead carry measured-subspace variance up to the noise level. Geometrically, the acquisition operator splits image content into a \emph{measured subspace} that the scanner pins down and an orthogonal, \emph{unmeasured} (null-space) subspace that the prior alone must fill in. A physically faithful posterior should therefore express variability only within the unmeasured subspace. This \emph{measured-subspace leakage}, posterior variability that instead spills into the measured subspace, implies that a portion of the reported ``uncertainty'' is physically inadmissible rather than informative (Fig.~\ref{fig:teaser}): it portrays sample-to-sample disagreement on coefficients that the scanner has already measured and that all samples condition on identically. Given that prior work has linked spurious or missing reconstruction details to potentially misleading interpretation~\cite{antun2020instabilities,bhadra2021hallucinations,morshuis2025sdr}, we treat such inadmissible variation as a property worth diagnosing and removing, while noting that its clinical impact is not evaluated here.

Why does this leakage arise? Existing posterior samplers all tie measurement fidelity directly to their internal update rules, and \emph{how} they do so creates a practical dilemma. Conceptually, they fall into two families: \emph{Soft} samplers (e.g., DPS~\cite{chung2023dps}, $\Pi$GDM~\cite{song2023pseudoinverse}), which guide the reverse dynamics via soft data gradients but lack a hard measured-subspace projection and are therefore prone to leakage, and \emph{Consistent} samplers (e.g., DDNM~\cite{wang2023ddnm}, Score-MRI~\cite{chung2022scoremri}), which build a hard measured-subspace projection~\cite{kawar2022ddrm} into their dynamics. Consequently, practitioners are forced to choose between the high perceptual quality of Soft samplers and the physical consistency of Consistent ones. Retrofitting a pre-trained Soft sampler with hard consistency is mathematically involved, typically requiring redesigned update equations and retuned hyperparameters rather than reuse of the model as-is.

To resolve this dilemma, we introduce measured-subspace consistency (MSC), a training-free, plug-and-play terminal correction that treats the base sampler as a black box and appends a classical measured-subspace consistency lock to its final output, reducing physically inadmissible measured-subspace variation without retraining or parameter retuning. In the ideal range/null-space setting MSC confines pairwise sample differences to the unmeasured subspace. In the practical multi-coil implementation it suppresses measured-subspace dispersion up to a bounded coil-combination residual, preserving a sampler's uncertainty about what the prior had to imagine while removing its disagreement about what the scanner already measured. Empirically, across six base samplers, two anatomies, and out-of-distribution priors, MSC sharply reduces leakage on Soft bases while remaining a near-no-op on Consistent ones, at negligible cost and without degrading PSNR/SSIM.

\noindent Our contributions are summarized as follows:
\begin{itemize}[noitemsep,topsep=2pt,parsep=0pt,partopsep=0pt]
\item We formalize \emph{measured-subspace leakage} as a physical-admissibility failure of diffusion posterior samplers and introduce two complementary k-space metrics, measured (MSD) and unmeasured (USD) subspace dispersion, that quantify whether reported uncertainty falls on already-acquired or genuinely unmeasured coefficients.
\item We repurpose the classical hard measured-subspace (range/null-space) data-consistency projection (not itself novel) as a plug-and-play, training-free \emph{output-placement} audit and correction for compatible image-space posterior samplers, reducing physically inadmissible measured-subspace variation without retraining or parameter retuning.
\item We mathematically characterize the transformation, deriving its exact agreement in the ideal case and a practical, sensitivity-weighted leakage bound that reveals the role of cross-subspace coupling.
\item We demonstrate the selective behavior of MSC across six base samplers, showing that it systematically reduces k-space leakage on Soft samplers while preserving unmeasured diversity and maintaining or improving reconstruction quality across multiple anatomies, masks, and priors.
\end{itemize}

%% file: sec/2_related.tex
\section{Related Work}
\label{sec:related}

\subsection{Diffusion Posterior Samplers: Soft vs. Consistent}
Posterior samplers differ in how they enforce data fidelity. \emph{Soft} samplers guide the diffusion prior toward the measurement with soft gradients or likelihood approximations: CSGM-Langevin~\cite{jalal2021robust} runs annealed Langevin dynamics with a linear measurement gradient, DPS~\cite{chung2023dps} and $\Pi$GDM~\cite{song2023pseudoinverse} steer general diffusion solvers using data term gradients, and Nila-DC~\cite{nila2024} dynamically scales a soft data-consistency gradient according to noise levels. While they yield high perceptual quality, their stochastic updates can perturb acquired k-space coefficients, resulting in measured-subspace leakage. Conversely, \emph{Consistent} samplers enforce fidelity inside the reverse update through different mechanisms: Score-MRI~\cite{chung2022scoremri,song2021solving} interleaves reverse SDE steps with POCS-style k-space replacement (the in-loop analogue of the terminal paste-back we use), DDRM~\cite{kawar2022ddrm} and DDNM~\cite{wang2023ddnm} apply an explicit SVD or range/null-space split to pin measured range components, and MCG~\cite{chung2022mcg} uses a manifold-constrained gradient correction that our audit finds effectively measured-consistent. Since these mechanisms are integrated in-loop, they cannot wrap arbitrary soft samplers without modifying their internal dynamics.

\subsection{MRI-Domain Solvers and Classical Projections}
Other methods enforce consistency inside specialized domains: SPIRiT-Diffusion~\cite{cao2024spirit} casts the self-consistency driven SPIRiT constraint~\cite{lustig2010spirit} as a k-space diffusion prior, DiffuseRecon~\cite{peng2022diffuserecon} injects k-space measurements during coarse-to-fine reverse diffusion, and MCLC~\cite{mclc2026} stabilizes latent-diffusion solvers with a measurement-consistent corrector. MSC instead operates in image space, controlling multi-coil dispersion without latent inversion. Closest to our decoupling of data consistency from the sampler's dynamics, the Decomposed Diffusion Sampler~\cite{chung2024dds} runs conjugate-gradient data consistency in a Krylov subspace outside the per-step denoiser, and ReSample~\cite{song2024resample} enforces hard data consistency for latent-diffusion solvers. MSC differs by applying a \emph{single terminal} multi-coil paste-back to a black-box sampler, with no inner CG/Krylov loop and no per-step interleaving. A related but orthogonal line is SDR~\cite{morshuis2025sdr}, whose diversity-seeking generation surfaces rare findings, whereas we remove physically inadmissible variation.

The projection operator itself is classical: it is a multi-coil POCS paste-back (POCSENSE~\cite{samsonov2004pocsense}), and relates more broadly to parallel-imaging methods that exploit coil structure to recover unsampled data (GRAPPA~\cite{griswold2002grappa}, SPIRiT~\cite{lustig2010spirit}). Applied once at the sampler output, MSC is a terminal multi-coil data-consistency paste-back. It should not be read as a new MRI projection operator. The contribution is orthogonal to the operator: (i) identifying measured-subspace variance as a failure mode of diffusion posterior samples rather than merely a reconstruction-error issue, (ii) introducing MSD/USD metrics that separate inadmissible measured variation from unmeasured posterior diversity, (iii) showing that a terminal classical projection is sufficient to correct Soft samplers without retraining or retuning and without collapsing null-space diversity, and (iv) explaining why the practical multi-coil realization leaves a bounded residual floor after sensitivity-weighted coil combination.

\subsection{Relation to Uncertainty Quantification}
Uncertainty quantification (UQ) for MRI reconstruction reports image-space uncertainty, either post hoc (VAE variance~\cite{edupuganti2021uncertainty}, Bayesian epistemic uncertainty~\cite{narnhofer2022bayesian}) or as the sample variance of the diffusion posterior samplers we study~\cite{jalal2021robust,chung2022scoremri}, and typically validates it against ground truth through calibration~\cite{guo2017calibration,kuleshov2018calibrated} or distribution-free conformal coverage, both for imaging inverse problems~\cite{angelopoulos2022im2im} and for diffusion models~\cite{teneggi2023trust}. MSC instead poses an orthogonal, ground-truth-free question: does the reported uncertainty vary coefficients the scanner already measured? Measured-subspace admissibility is in fact a necessary condition for a calibrated posterior (samples conditioned on the same fixed measurement cannot legitimately disagree on it), yet unlike calibration it needs no ground truth or held-out data. MSC is therefore not itself a UQ method but a physical-admissibility check on posterior dispersion that complements these image-space methods. Its MSD/USD metrics could in turn audit the validity of their variance maps.

%% file: sec/3_method.tex
\section{Method}
\label{sec:method}

\begin{figure*}[t]
\centering
\includegraphics[width=0.98\linewidth]{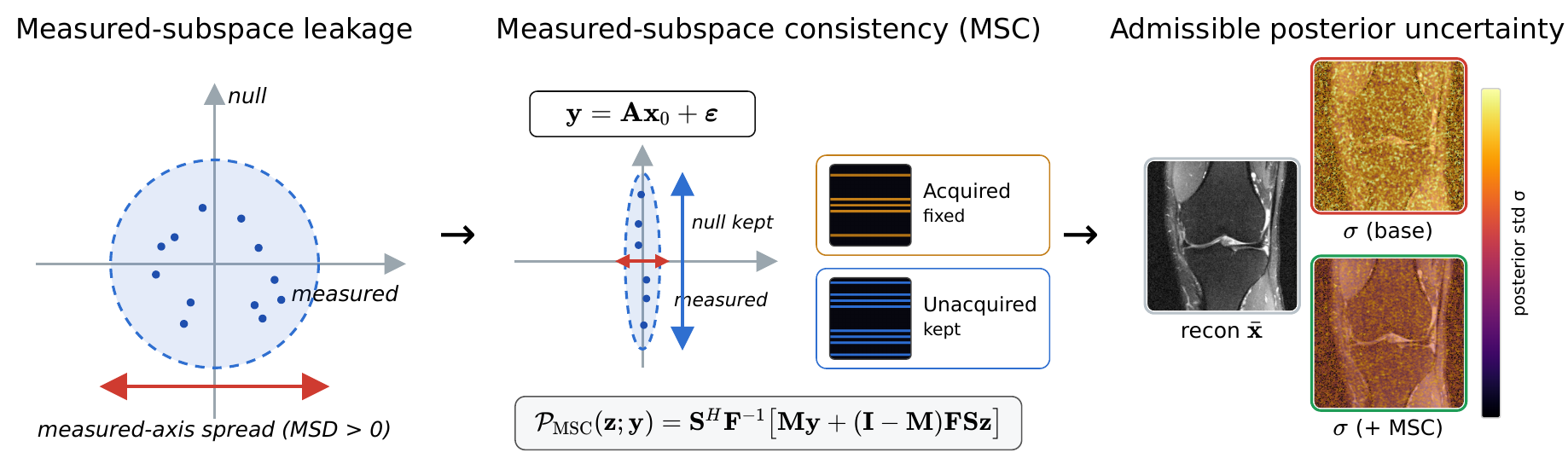}
\caption{\textbf{Overview of measured-subspace consistency.} Left: leakage is posterior spread along measured coordinates. Middle: MSC replaces acquired multi-coil k-space coefficients of proposal $\mathbf{z}$ with observations $\mathbf{y}$, keeps unacquired coefficients, and maps back with $\mathbf{S}^{H}\mathbf{F}^{-1}$. Right: MSD/USD audit the resulting measured and unmeasured dispersion.}
\label{fig:method}
\end{figure*}

\subsection{Measured-Subspace Leakage Problem}
\label{sec:prelim}
We consider 2D multi-coil Cartesian accelerated MRI, where unitary Fourier transform $\mathbf{F}$ and binary mask $\mathbf{M}$ induce an orthogonal acquired/unacquired split. Non-Cartesian trajectories are left to future work. With $C$ sensitivity-weighted coils, the forward model is
\begin{equation}
\y=\mathbf{A}\x+\boldsymbol{\varepsilon}, \qquad \mathbf{A}=\mathbf{M}\mathbf{F}\mathbf{S},
\label{eq:forward}
\end{equation}
where $\x\in\mathbb{C}^{N}$ is the complex image, $\mathbf{S}$ applies coil sensitivities, $\boldsymbol{\varepsilon}$ is acquisition noise, and $\y=\mathbf{M}\y$ is zero-filled acquired multi-coil k-space. The adjoint $\mathbf{A}^{H}=\mathbf{S}^{H}\mathbf{F}^{-1}\mathbf{M}$ maps k-space to image space by inverse Fourier transform and sensitivity-weighted coil combination. Since $\mathbf{M}$ discards unsampled locations, $\mathbf{A}$ is rank-deficient: sample differences split into a measured component constrained by the scanner and a null-space component $\mathcal{N}(\mathbf{A}) = \{\x : \mathbf{M}\mathbf{F}\mathbf{S}\x = \mathbf{0}\}$ about which the measurement is silent. In this hard-constraint (noise-negligible) regime, only the latter represents legitimate posterior uncertainty. Under measurement noise the measured component additionally admits variance up to the noise level.

Against this model, diffusion posterior samplers denoise from noise while guiding each image-space iterate $\x_t$ toward $\y$, either through a \emph{soft} data term or a \emph{hard} measured-subspace projection. Running the stochastic reverse process $L$ times yields samples $\{\x^{(\ell)}\}_{\ell=1}^{L}$ whose spread is commonly reported as uncertainty.

Not all dispersion is physically justified. Variation on unacquired coefficients reflects ambiguity. Variation on acquired coefficients means samples disagree on recorded data, which we term measured-subspace leakage. To quantify it, let $k^{(\ell)}_{c,i}$ be the multi-coil k-space coefficient of sample $\ell$ at coil $c$ and position $i$, and define the complex sample standard deviation across chains as
\begin{equation}
\operatorname{Std}_{\ell}\!\big[k^{(\ell)}_{c,i}\big]=\Big(\tfrac{1}{L-1}\sum_{\ell=1}^{L}\big|k^{(\ell)}_{c,i}-\bar{k}_{c,i}\big|^{2}\Big)^{1/2},
\quad \bar{k}_{c,i}=\tfrac{1}{L}\sum_{\ell=1}^{L}k^{(\ell)}_{c,i}.
\label{eq:std}
\end{equation}
We then define the \emph{measured-subspace dispersion} (MSD, denoted $\vmeas$) and \emph{unmeasured-subspace dispersion} (USD, denoted $\vunmeas$) as the mean complex standard deviation over the measured ($M_i=1$) and unmeasured ($M_i=0$) regions, respectively:
\begin{equation}
\begin{aligned}
\vmeas &= \operatorname{Avg}_{c, i: M_i=1} \operatorname{Std}_{\ell}\!\left[k^{(\ell)}_{c,i}\right], \\
\vunmeas &= \operatorname{Avg}_{c, i: M_i=0} \operatorname{Std}_{\ell}\!\left[k^{(\ell)}_{c,i}\right].
\end{aligned}
\end{equation}
A measured-subspace consistent sampler should have low MSD without suppressing USD. Reconstruction quality is evaluated separately with NMSE, PSNR, and SSIM~\cite{wang2004ssim} on the chain-mean magnitude after one least-squares scale match. Two details matter. First, MSD/USD partition multi-coil k-space into acquired/unacquired regions. In the single-coil unitary case this coincides with the image-space range/null-space split $\mathcal{R}(\mathbf{A}^{H})/\mathcal{N}(\mathbf{A})$. Second, dispersion is computed on the \emph{final returned} iterate by re-encoding with $\mathbf{F}\mathbf{S}$. For MSC, this includes the cross-subspace coupling floor of Prop.~\ref{prop:bound}. We therefore restrict cross-sampler MSD comparisons to scale-invariant ratios, while absolute MSD ($\times10^{-4}$) remains comparable within a sampler/contrast.

\subsection{Ideal Consistency Operator}
MSC uses a training-free terminal lock to audit and correct samplers that expose an image-space output. We model a compatible base sampler as a map $\mathbf{z}_{t-1} = S_t(\x_t; \y)$ (Figure~\ref{fig:method}), where $S_t$ may use soft data guidance or in-loop hard projection. Compatibility only requires an exposed image-space iterate, an applicable forward operator $\mathbf{A}$, and the ability to append a final projection. K-space-native, latent-only, or exact-likelihood samplers may still be diagnosed with MSD/USD, but are outside the wrapper interface.

To establish the mathematical foundation of our method, we first analyze measured-subspace consistency in the ideal, range-consistent setting ($\y \in \mathcal{R}(\mathbf{A})$). In this case, the consistency operator replaces the measured component of the iterate with the acquired scanner data:
\begin{equation}
\Pideal(\mathbf{z};\y)=\mathbf{z}+\mathbf{A}^{+}(\y-\mathbf{A}\mathbf{z}),
\label{eq:pstar}
\end{equation}
with $\mathbf{A}^{+}$ the pseudoinverse of $\mathbf{A}=\mathbf{M}\mathbf{F}\mathbf{S}$. $\Pideal$ overwrites the measured component with $\y$ and leaves the null-space component untouched. This is the same classical range/null-space principle used by hard data-consistency and null-space diffusion solvers~\cite{kawar2022ddrm,wang2023ddnm}. Here it is used at the output to analyze and correct posterior leakage.

\begin{proposition}[Exact measured-subspace agreement]
\label{prop:exact}
For any iterates $\mathbf{z}_1,\mathbf{z}_2$, the transform \eqref{eq:pstar} satisfies
\begin{equation}
\Pideal(\mathbf{z}_1;\y)-\Pideal(\mathbf{z}_2;\y)=(\mathbf{I}-\mathbf{A}^{+}\mathbf{A})(\mathbf{z}_1-\mathbf{z}_2)\in\mathcal{N}(\mathbf{A}),
\label{eq:propexact}
\end{equation}
so that $\mathbf{A}\big(\Pideal(\mathbf{z}_1;\y)-\Pideal(\mathbf{z}_2;\y)\big)=\mathbf{0}$. Moreover, $\mathbf{A}\mathbf{z}=\y$ implies $\Pideal(\mathbf{z};\y)=\mathbf{z}$.
\end{proposition}
\begin{proof}
Subtracting two instances of \cref{eq:pstar}, the data terms $\mathbf{A}^{+}\y$ cancel, leaving $(\mathbf{I}-\mathbf{A}^{+}\mathbf{A})(\mathbf{z}_1-\mathbf{z}_2)$. Since $\mathbf{A}(\mathbf{I}-\mathbf{A}^{+}\mathbf{A})=\mathbf{A}-\mathbf{A}\mathbf{A}^{+}\mathbf{A}=\mathbf{0}$, this difference lies in $\mathcal{N}(\mathbf{A})$. Finally, if $\mathbf{A}\mathbf{z}=\y$ then $\y-\mathbf{A}\mathbf{z}=\mathbf{0}$, so $\Pideal(\mathbf{z};\y)=\mathbf{z}$.
\end{proof}
Proposition~\ref{prop:exact} implies zero ideal MSD with preserved null-space diversity, and the fixed-point property leaves already consistent samples unchanged. This explains why MSC should suppress Soft-sampler leakage while acting as a near-no-op on Consistent bases, up to the practical residual below.

\begin{table*}[t]
\centering
\caption{\textbf{In-domain knee reconstruction.} Single PD knee prior across masks, accelerations, and PD/PDFS contrasts. Metrics: MSD ($\times10^{-4}$, $\downarrow$), PSNR (dB, $\uparrow$), SSIM ($\uparrow$). Better of base/+MSC in \textbf{bold} (display ties unbolded).}
\label{tab:knee}
\scriptsize
\renewcommand{\arraystretch}{0.85}
\setlength{\tabcolsep}{1.1pt}
\begin{tabular}{l|c|l|*{12}{c}|*{12}{c}}
\toprule
\textbf{Sampling Patterns} & $R$ & Metric & \multicolumn{12}{c|}{\textbf{PD}} & \multicolumn{12}{c}{\textbf{PDFS}}\\
\cmidrule(lr){4-15}\cmidrule(lr){16-27}
& & & \multicolumn{6}{c}{\emph{Soft}} & \multicolumn{6}{c|}{\emph{Consistent}} & \multicolumn{6}{c}{\emph{Soft}} & \multicolumn{6}{c}{\emph{Consistent}} \\
\cmidrule(lr){4-9}\cmidrule(lr){10-15}\cmidrule(lr){16-21}\cmidrule(lr){22-27}
& & & \multicolumn{2}{c}{DPS} & \multicolumn{2}{c}{Nila-DC} & \multicolumn{2}{c}{$\Pi$GDM} & \multicolumn{2}{c}{DDNM} & \multicolumn{2}{c}{MCG} & \multicolumn{2}{c|}{Score} & \multicolumn{2}{c}{DPS} & \multicolumn{2}{c}{Nila-DC} & \multicolumn{2}{c}{$\Pi$GDM} & \multicolumn{2}{c}{DDNM} & \multicolumn{2}{c}{MCG} & \multicolumn{2}{c}{Score} \\
\cmidrule(lr){4-5}\cmidrule(lr){6-7}\cmidrule(lr){8-9}\cmidrule(lr){10-11}\cmidrule(lr){12-13}\cmidrule(lr){14-15}\cmidrule(lr){16-17}\cmidrule(lr){18-19}\cmidrule(lr){20-21}\cmidrule(lr){22-23}\cmidrule(lr){24-25}\cmidrule(lr){26-27}
& & & base & MSC & base & MSC & base & MSC & base & MSC & base & MSC & base & \multicolumn{1}{c|}{MSC} & base & MSC & base & MSC & base & MSC & base & MSC & base & MSC & base & MSC \\
\midrule
\multirow{9}{*}{Uniform 1D} & \multirow{3}{*}{$4\times$} & MSD & {\textcolor{black!55}{11.6}} & \textbf{4.9} & {\textcolor{black!55}{23.1}} & \textbf{8.6} & {\textcolor{black!55}{42.4}} & \textbf{13.6} & {\textcolor{black!55}{7.3}} & \textbf{5.6} & {\textcolor{black!55}{4.5}} & \textbf{3.8} & {\textcolor{black!55}{5.5}} & \textbf{4.5} & {\textcolor{black!55}{13.9}} & \textbf{7.1} & {\textcolor{black!55}{22.0}} & \textbf{7.9} & {\textcolor{black!55}{37.9}} & \textbf{9.6} & {\textcolor{black!55}{9.2}} & \textbf{7.1} & {\textcolor{black!55}{6.1}} & \textbf{5.2} & {\textcolor{black!55}{8.0}} & \textbf{6.7} \\
 & & PSNR & {\textcolor{black!55}{37.4}} & \textbf{37.5} & {\textcolor{black!55}{33.6}} & \textbf{35.0} & {\textcolor{black!55}{24.6}} & \textbf{27.2} & {\textcolor{black!55}{36.4}} & \textbf{36.6} & 37.0 & 37.0 & {\textcolor{black!55}{37.1}} & \textbf{37.2} & {\textcolor{black!55}{33.1}} & \textbf{33.5} & {\textcolor{black!55}{31.7}} & \textbf{33.2} & {\textcolor{black!55}{23.6}} & \textbf{27.3} & \textbf{33.4} & {\textcolor{black!55}{33.3}} & \textbf{33.4} & {\textcolor{black!55}{33.2}} & \textbf{33.3} & {\textcolor{black!55}{33.1}} \\
 & & SSIM & .94 & .94 & {\textcolor{black!55}{.85}} & \textbf{.91} & {\textcolor{black!55}{.66}} & \textbf{.78} & .93 & .93 & .93 & .93 & .94 & .94 & {\textcolor{black!55}{.81}} & \textbf{.84} & {\textcolor{black!55}{.75}} & \textbf{.83} & {\textcolor{black!55}{.52}} & \textbf{.72} & \textbf{.84} & {\textcolor{black!55}{.83}} & \textbf{.84} & .83 & .83 & {\textcolor{black!55}{.83}} \\
\cmidrule(lr){2-27}
 & \multirow{3}{*}{$8\times$} & MSD & {\textcolor{black!55}{30.9}} & \textbf{4.7} & {\textcolor{black!55}{26.8}} & \textbf{8.2} & {\textcolor{black!55}{57.4}} & \textbf{14.4} & {\textcolor{black!55}{6.7}} & \textbf{5.2} & {\textcolor{black!55}{4.7}} & \textbf{4.0} & {\textcolor{black!55}{5.3}} & \textbf{4.4} & {\textcolor{black!55}{12.0}} & \textbf{6.3} & {\textcolor{black!55}{25.0}} & \textbf{7.6} & {\textcolor{black!55}{52.0}} & \textbf{10.1} & {\textcolor{black!55}{7.2}} & \textbf{5.8} & {\textcolor{black!55}{5.3}} & \textbf{4.7} & {\textcolor{black!55}{6.3}} & \textbf{5.3} \\
 & & PSNR & 30.9 & 30.9 & {\textcolor{black!55}{29.6}} & \textbf{29.8} & {\textcolor{black!55}{22.8}} & \textbf{24.0} & 29.9 & 29.9 & 29.6 & 29.6 & 29.7 & 29.7 & 26.6 & 26.6 & {\textcolor{black!55}{26.5}} & \textbf{26.7} & {\textcolor{black!55}{22.8}} & \textbf{24.8} & {\textcolor{black!55}{26.1}} & \textbf{26.2} & \textbf{26.6} & 26.6 & 26.6 & {\textcolor{black!55}{26.6}} \\
 & & SSIM & .85 & .85 & {\textcolor{black!55}{.80}} & \textbf{.83} & {\textcolor{black!55}{.62}} & \textbf{.67} & .83 & .83 & .83 & .83 & .83 & .83 & {\textcolor{black!55}{.65}} & \textbf{.66} & {\textcolor{black!55}{.61}} & \textbf{.66} & {\textcolor{black!55}{.49}} & \textbf{.61} & {\textcolor{black!55}{.66}} & \textbf{.67} & .68 & .68 & .68 & .68 \\
\cmidrule(lr){2-27}
 & \multirow{3}{*}{$12\times$} & MSD & {\textcolor{black!55}{30.7}} & \textbf{6.4} & {\textcolor{black!55}{32.4}} & \textbf{11.8} & {\textcolor{black!55}{66.9}} & \textbf{20.1} & {\textcolor{black!55}{8.9}} & \textbf{7.1} & {\textcolor{black!55}{6.7}} & \textbf{5.9} & {\textcolor{black!55}{7.2}} & \textbf{6.2} & {\textcolor{black!55}{14.8}} & \textbf{8.5} & {\textcolor{black!55}{28.9}} & \textbf{10.8} & {\textcolor{black!55}{59.0}} & \textbf{15.5} & {\textcolor{black!55}{9.2}} & \textbf{7.4} & {\textcolor{black!55}{7.2}} & \textbf{6.4} & {\textcolor{black!55}{8.2}} & \textbf{7.2} \\
 & & PSNR & 28.4 & 28.4 & {\textcolor{black!55}{27.3}} & \textbf{27.4} & {\textcolor{black!55}{21.0}} & \textbf{21.6} & 27.0 & 27.0 & 26.2 & 26.2 & 26.8 & 26.8 & 24.2 & 24.2 & {\textcolor{black!55}{24.3}} & \textbf{24.4} & {\textcolor{black!55}{22.0}} & \textbf{23.2} & 24.3 & 24.3 & \textbf{24.7} & 24.7 & 24.7 & {\textcolor{black!55}{24.7}} \\
 & & SSIM & .81 & .81 & {\textcolor{black!55}{.76}} & \textbf{.78} & {\textcolor{black!55}{.57}} & \textbf{.60} & .78 & .78 & .76 & .76 & {\textcolor{black!55}{.77}} & \textbf{.78} & .60 & .60 & {\textcolor{black!55}{.56}} & \textbf{.60} & {\textcolor{black!55}{.46}} & \textbf{.55} & .61 & .61 & {\textcolor{black!55}{.62}} & \textbf{.63} & .63 & .63 \\
\midrule
\multirow{6}{*}{Gaussian 1D} & \multirow{3}{*}{$8\times$} & MSD & {\textcolor{black!55}{34.8}} & \textbf{3.9} & {\textcolor{black!55}{28.0}} & \textbf{7.7} & {\textcolor{black!55}{61.4}} & \textbf{14.9} & {\textcolor{black!55}{5.8}} & \textbf{4.5} & {\textcolor{black!55}{3.9}} & \textbf{3.3} & {\textcolor{black!55}{4.3}} & \textbf{3.6} & {\textcolor{black!55}{11.8}} & \textbf{5.6} & {\textcolor{black!55}{28.1}} & \textbf{7.6} & {\textcolor{black!55}{60.0}} & \textbf{12.3} & {\textcolor{black!55}{6.9}} & \textbf{5.4} & {\textcolor{black!55}{4.7}} & \textbf{4.1} & {\textcolor{black!55}{5.5}} & \textbf{4.7} \\
 & & PSNR & 33.2 & 33.2 & {\textcolor{black!55}{31.9}} & \textbf{32.5} & {\textcolor{black!55}{24.5}} & \textbf{26.6} & 32.7 & 32.7 & 32.8 & 32.8 & {\textcolor{black!55}{32.9}} & \textbf{33.0} & {\textcolor{black!55}{31.4}} & \textbf{31.5} & {\textcolor{black!55}{30.6}} & \textbf{31.2} & {\textcolor{black!55}{23.6}} & \textbf{26.6} & 31.4 & 31.4 & 31.5 & 31.5 & \textbf{31.6} & {\textcolor{black!55}{31.5}} \\
 & & SSIM & .88 & .88 & {\textcolor{black!55}{.84}} & \textbf{.87} & {\textcolor{black!55}{.66}} & \textbf{.74} & .88 & .88 & .88 & .88 & .88 & .88 & {\textcolor{black!55}{.77}} & \textbf{.78} & {\textcolor{black!55}{.72}} & \textbf{.77} & {\textcolor{black!55}{.51}} & \textbf{.67} & .78 & .78 & .78 & .78 & .78 & .78 \\
\cmidrule(lr){2-27}
 & \multirow{3}{*}{$12\times$} & MSD & {\textcolor{black!55}{28.7}} & \textbf{5.5} & {\textcolor{black!55}{34.3}} & \textbf{11.9} & {\textcolor{black!55}{76.5}} & \textbf{24.0} & {\textcolor{black!55}{8.5}} & \textbf{6.5} & {\textcolor{black!55}{5.7}} & \textbf{4.9} & {\textcolor{black!55}{6.4}} & \textbf{5.4} & {\textcolor{black!55}{14.1}} & \textbf{8.1} & {\textcolor{black!55}{32.8}} & \textbf{12.0} & {\textcolor{black!55}{68.4}} & \textbf{19.9} & {\textcolor{black!55}{9.5}} & \textbf{7.6} & {\textcolor{black!55}{6.7}} & \textbf{5.9} & {\textcolor{black!55}{7.8}} & \textbf{6.7} \\
 & & PSNR & 31.8 & 31.8 & {\textcolor{black!55}{31.0}} & \textbf{31.2} & {\textcolor{black!55}{22.3}} & \textbf{23.6} & 31.3 & 31.3 & 31.3 & 31.3 & 31.5 & 31.5 & 29.6 & 29.6 & {\textcolor{black!55}{29.2}} & \textbf{29.6} & {\textcolor{black!55}{22.6}} & \textbf{24.3} & 29.3 & 29.3 & 29.4 & 29.4 & \textbf{29.6} & {\textcolor{black!55}{29.5}} \\
 & & SSIM & .86 & .86 & {\textcolor{black!55}{.82}} & \textbf{.84} & {\textcolor{black!55}{.61}} & \textbf{.66} & .85 & .85 & .85 & .85 & .85 & .85 & {\textcolor{black!55}{.71}} & \textbf{.72} & {\textcolor{black!55}{.69}} & \textbf{.72} & {\textcolor{black!55}{.48}} & \textbf{.60} & {\textcolor{black!55}{.72}} & \textbf{.73} & .73 & .73 & .73 & .73 \\
\midrule
\multirow{6}{*}{Gaussian 2D} & \multirow{3}{*}{$8\times$} & MSD & {\textcolor{black!55}{24.0}} & \textbf{5.9} & {\textcolor{black!55}{32.7}} & \textbf{15.1} & {\textcolor{black!55}{77.1}} & \textbf{29.6} & {\textcolor{black!55}{10.0}} & \textbf{7.3} & {\textcolor{black!55}{6.0}} & \textbf{4.8} & {\textcolor{black!55}{7.1}} & \textbf{5.6} & {\textcolor{black!55}{14.9}} & \textbf{9.0} & {\textcolor{black!55}{34.1}} & \textbf{15.2} & {\textcolor{black!55}{70.7}} & \textbf{23.4} & {\textcolor{black!55}{12.1}} & \textbf{9.0} & {\textcolor{black!55}{7.5}} & \textbf{6.2} & {\textcolor{black!55}{9.2}} & \textbf{7.4} \\
 & & PSNR & {\textcolor{black!55}{35.8}} & \textbf{35.9} & {\textcolor{black!55}{33.7}} & \textbf{34.8} & {\textcolor{black!55}{26.2}} & \textbf{31.4} & {\textcolor{black!55}{35.5}} & \textbf{35.6} & 35.7 & 35.7 & 35.8 & 35.8 & \textbf{33.6} & {\textcolor{black!55}{33.5}} & {\textcolor{black!55}{32.3}} & \textbf{33.2} & {\textcolor{black!55}{24.3}} & \textbf{30.3} & \textbf{33.6} & {\textcolor{black!55}{33.5}} & \textbf{33.6} & 33.5 & 33.5 & {\textcolor{black!55}{33.4}} \\
 & & SSIM & .92 & .92 & {\textcolor{black!55}{.86}} & \textbf{.90} & {\textcolor{black!55}{.71}} & \textbf{.86} & {\textcolor{black!55}{.91}} & \textbf{.92} & .92 & .92 & .92 & .92 & {\textcolor{black!55}{.82}} & \textbf{.84} & {\textcolor{black!55}{.76}} & \textbf{.82} & {\textcolor{black!55}{.56}} & \textbf{.78} & .84 & .84 & \textbf{.83} & .83 & .83 & {\textcolor{black!55}{.83}} \\
\cmidrule(lr){2-27}
 & \multirow{3}{*}{$15\times$} & MSD & {\textcolor{black!55}{31.4}} & \textbf{10.4} & {\textcolor{black!55}{43.3}} & \textbf{25.5} & {\textcolor{black!55}{121.2}} & \textbf{62.2} & {\textcolor{black!55}{16.8}} & \textbf{13.3} & {\textcolor{black!55}{11.2}} & \textbf{9.5} & {\textcolor{black!55}{12.7}} & \textbf{10.7} & {\textcolor{black!55}{22.7}} & \textbf{15.5} & {\textcolor{black!55}{47.5}} & \textbf{27.8} & {\textcolor{black!55}{106.6}} & \textbf{47.8} & {\textcolor{black!55}{20.2}} & \textbf{16.2} & {\textcolor{black!55}{14.1}} & \textbf{12.2} & {\textcolor{black!55}{16.2}} & \textbf{14.0} \\
 & & PSNR & {\textcolor{black!55}{34.2}} & \textbf{34.3} & {\textcolor{black!55}{33.0}} & \textbf{33.6} & {\textcolor{black!55}{24.4}} & \textbf{27.7} & {\textcolor{black!55}{33.9}} & \textbf{34.0} & 34.1 & 34.1 & 34.2 & 34.2 & 32.1 & 32.1 & {\textcolor{black!55}{31.5}} & \textbf{31.9} & {\textcolor{black!55}{23.4}} & \textbf{26.8} & \textbf{32.1} & 32.1 & 32.1 & {\textcolor{black!55}{32.0}} & \textbf{31.9} & {\textcolor{black!55}{31.8}} \\
 & & SSIM & .90 & .90 & {\textcolor{black!55}{.85}} & \textbf{.87} & {\textcolor{black!55}{.66}} & \textbf{.77} & {\textcolor{black!55}{.89}} & \textbf{.90} & .90 & .90 & .90 & .90 & .80 & .80 & {\textcolor{black!55}{.75}} & \textbf{.78} & {\textcolor{black!55}{.51}} & \textbf{.68} & \textbf{.81} & .80 & .80 & .80 & .80 & {\textcolor{black!55}{.79}} \\
\midrule
\multirow{6}{*}{Poisson 2D} & \multirow{3}{*}{$8\times$} & MSD & {\textcolor{black!55}{24.5}} & \textbf{8.6} & {\textcolor{black!55}{30.8}} & \textbf{18.0} & {\textcolor{black!55}{70.6}} & \textbf{34.3} & {\textcolor{black!55}{13.1}} & \textbf{10.2} & {\textcolor{black!55}{8.5}} & \textbf{7.1} & {\textcolor{black!55}{10.1}} & \textbf{8.3} & {\textcolor{black!55}{18.5}} & \textbf{12.5} & {\textcolor{black!55}{32.4}} & \textbf{18.5} & {\textcolor{black!55}{62.8}} & \textbf{25.9} & {\textcolor{black!55}{15.6}} & \textbf{12.2} & {\textcolor{black!55}{11.0}} & \textbf{9.3} & {\textcolor{black!55}{14.1}} & \textbf{11.8} \\
 & & PSNR & {\textcolor{black!55}{36.5}} & \textbf{36.8} & {\textcolor{black!55}{33.6}} & \textbf{34.8} & {\textcolor{black!55}{25.7}} & \textbf{30.1} & {\textcolor{black!55}{36.0}} & \textbf{36.2} & {\textcolor{black!55}{36.5}} & \textbf{36.6} & {\textcolor{black!55}{36.7}} & \textbf{36.8} & \textbf{33.7} & {\textcolor{black!55}{33.6}} & {\textcolor{black!55}{32.2}} & \textbf{33.1} & {\textcolor{black!55}{24.3}} & \textbf{29.5} & 33.6 & 33.6 & \textbf{33.5} & 33.3 & 33.3 & {\textcolor{black!55}{33.0}} \\
 & & SSIM & {\textcolor{black!55}{.92}} & \textbf{.93} & {\textcolor{black!55}{.86}} & \textbf{.90} & {\textcolor{black!55}{.69}} & \textbf{.83} & .92 & .92 & .93 & .93 & .93 & .93 & {\textcolor{black!55}{.82}} & \textbf{.84} & {\textcolor{black!55}{.76}} & \textbf{.82} & {\textcolor{black!55}{.55}} & \textbf{.75} & \textbf{.84} & .83 & .83 & .82 & .82 & {\textcolor{black!55}{.81}} \\
\cmidrule(lr){2-27}
 & \multirow{3}{*}{$15\times$} & MSD & {\textcolor{black!55}{34.7}} & \textbf{11.6} & {\textcolor{black!55}{39.2}} & \textbf{24.6} & {\textcolor{black!55}{104.6}} & \textbf{55.1} & {\textcolor{black!55}{17.1}} & \textbf{13.7} & {\textcolor{black!55}{11.9}} & \textbf{10.1} & {\textcolor{black!55}{13.3}} & \textbf{11.3} & {\textcolor{black!55}{23.4}} & \textbf{15.4} & {\textcolor{black!55}{42.1}} & \textbf{26.2} & {\textcolor{black!55}{99.6}} & \textbf{44.3} & {\textcolor{black!55}{19.4}} & \textbf{15.8} & {\textcolor{black!55}{14.1}} & \textbf{12.2} & {\textcolor{black!55}{16.7}} & \textbf{14.4} \\
 & & PSNR & {\textcolor{black!55}{34.9}} & \textbf{35.1} & {\textcolor{black!55}{33.0}} & \textbf{33.6} & {\textcolor{black!55}{24.9}} & \textbf{28.0} & {\textcolor{black!55}{34.2}} & \textbf{34.4} & 34.7 & 34.7 & {\textcolor{black!55}{34.9}} & \textbf{35.0} & {\textcolor{black!55}{32.3}} & \textbf{32.4} & {\textcolor{black!55}{31.5}} & \textbf{32.0} & {\textcolor{black!55}{23.8}} & \textbf{27.7} & {\textcolor{black!55}{32.2}} & \textbf{32.3} & \textbf{32.4} & 32.3 & 32.3 & {\textcolor{black!55}{32.2}} \\
 & & SSIM & .90 & .90 & {\textcolor{black!55}{.84}} & \textbf{.87} & {\textcolor{black!55}{.66}} & \textbf{.77} & {\textcolor{black!55}{.89}} & \textbf{.90} & .90 & .90 & .90 & .90 & .80 & .80 & {\textcolor{black!55}{.74}} & \textbf{.78} & {\textcolor{black!55}{.52}} & \textbf{.68} & .81 & .81 & \textbf{.80} & .80 & .80 & {\textcolor{black!55}{.80}} \\
\bottomrule
\end{tabular}
\end{table*}

\subsection{Practical Multi-Coil Operator}
Although Proposition~\ref{prop:exact} characterizes the ideal projection, $\Pideal$ is impractical for multi-coil MRI. Computing the exact SENSE pseudoinverse $\mathbf{A}^+$ is computationally expensive and numerically unstable in low-sensitivity regions. We therefore use a practical multi-coil consistency lock $\Pmsc$ realized via sensitivity-weighted coil combination, the diffusion-output analogue of classical multi-coil POCS data consistency~\cite{samsonov2004pocsense}:
\begin{equation}
\Pmsc(\mathbf{z};\y)=\mathbf{S}^{H}\mathbf{F}^{-1}\!\left[\mathbf{M}\y+(\mathbf{I}-\mathbf{M})\mathbf{F}\mathbf{S}\mathbf{z}\right],
\label{eq:pmsc}
\end{equation}
The bracketed k-space exactly matches acquired entries before coil combination. After applying $\mathbf{S}^{H}$ and re-encoding, the match becomes approximate because coil sensitivities couple subspaces. We bound this residual leakage next.
 
\begin{proposition}[Practical leakage bound]
\label{prop:bound}
For the operator \eqref{eq:pmsc}, let $\Delta\mathbf{z}=\mathbf{z}_1-\mathbf{z}_2$ and $\mathbf{u}=(\mathbf{I}-\mathbf{M})\mathbf{F}\mathbf{S}\Delta\mathbf{z}$ be the unmeasured multi-coil k-space difference retained by the lock. Then
\begin{equation}
\Pmsc(\mathbf{z}_1;\y)-\Pmsc(\mathbf{z}_2;\y)=\mathbf{S}^{H}\mathbf{F}^{-1}\mathbf{u},
\end{equation}
and re-encoding gives $\mathbf{A}\big(\Pmsc(\mathbf{z}_1;\y)-\Pmsc(\mathbf{z}_2;\y)\big)=\mathbf{C}_{\mathrm{cross}}\mathbf{u}$ with $\mathbf{C}_{\mathrm{cross}}=\mathbf{M}\mathbf{F}\mathbf{S}\mathbf{S}^{H}\mathbf{F}^{-1}$. Hence
\begin{equation}
\begin{aligned}
&\big\|\mathbf{A}\big(\Pmsc(\mathbf{z}_1;\y)-\Pmsc(\mathbf{z}_2;\y)\big)\big\|_2 \\
&\quad \le \|\mathbf{C}_{\mathrm{cross}}\|_2 \big\|(\mathbf{I}-\mathbf{M})\mathbf{F}\mathbf{S}(\mathbf{z}_1-\mathbf{z}_2)\big\|_2.
\end{aligned}
\label{eq:bound}
\end{equation}
\end{proposition}
\begin{proof}
Subtracting two outputs of the linear operator in \cref{eq:pmsc} cancels the measured term:
\begin{align*}
\Pmsc(\mathbf{z}_1;\y)-\Pmsc(\mathbf{z}_2;\y)
&= \mathbf{S}^{H}\mathbf{F}^{-1}(\mathbf{I}-\mathbf{M})\mathbf{F}\mathbf{S}\Delta\mathbf{z} \\
&= \mathbf{S}^{H}\mathbf{F}^{-1}\mathbf{u}.
\end{align*}
Left-multiplying by $\mathbf{A}=\mathbf{M}\mathbf{F}\mathbf{S}$ gives
\begin{align*}
\mathbf{A}\big(\Pmsc(\mathbf{z}_1;\y)-\Pmsc(\mathbf{z}_2;\y)\big)
&= \mathbf{M}\mathbf{F}\mathbf{S}\mathbf{S}^{H}\mathbf{F}^{-1}\mathbf{u} \\
&= \mathbf{C}_{\mathrm{cross}}\mathbf{u}.
\end{align*}
The stated inequality follows from submultiplicativity of the operator norm and the definition of $\mathbf{u}$.
\end{proof}
Physically, because the measurement terms cancel, post-MSC measured dispersion is independent of the base sampler's original measured-subspace disagreement. It arises only from unmeasured-to-measured cross-coupling through $\mathbf{S}^{H}$, captured by $\mathbf{C}_{\mathrm{cross}}=\mathbf{M}\mathbf{F}\mathbf{S}\mathbf{S}^{H}\mathbf{F}^{-1}$. This residual is intrinsic to sensitivity-weighted coil combination: ESPIRiT maps are root-sum-of-squares normalized, so $\mathbf{S}^{H}\mathbf{S}=\mathbf{I}$, but for $C>1$ coils $\mathbf{S}\mathbf{S}^{H}$ is then a rank-$N$ orthogonal projector on the $CN$-dimensional multi-coil space (not the identity), so $\mathbf{C}_{\mathrm{cross}}\neq\mathbf{0}$ even for perfectly conditioned maps and MSC reduces measured dispersion strongly but not exactly to zero. The same normalization gives $\mathbf{S}^{+}=(\mathbf{S}^{H}\mathbf{S})^{-1}\mathbf{S}^{H}=\mathbf{S}^{H}$, so substituting the SENSE pseudoinverse~\cite{pruessmann1999sense,uecker2014espirit} leaves post-MSC dispersion essentially unchanged and cannot lower a floor set by $\mathbf{S}\mathbf{S}^{H}$. Reducing it would instead require a regularized inverse, which can amplify low-sensitivity noise, so we use $\mathbf{S}^{H}$ by default.

In practice, we apply $\Pmsc$ once, to the sampler's final output: $\x_{\mathrm{out}}\!\leftarrow\!\Pmsc(\x_{\mathrm{out}}^{\mathrm{base}};\y)$ (\textbf{MSC}). It has no trainable parameters, requires no retuning, and costs one forward and inverse Fourier transform.

%% file: sec/4_experiments.tex
\begin{table*}[t]
\centering
\caption{\textbf{OOD brain reconstruction.} Knee prior reconstructs five fastMRI brain contrasts at $R=8,12$. Metrics: MSD ($\times10^{-4}$, $\downarrow$), PSNR (dB, $\uparrow$), SSIM ($\uparrow$). Better of base/+MSC in \textbf{bold} (display ties unbolded).}
\label{tab:brain}
\footnotesize
\renewcommand{\arraystretch}{0.85}
\setlength{\tabcolsep}{1.1pt}
\begin{tabular}{l|l|*{12}{c}|*{12}{c}}
\toprule
\textbf{Contrast} & Metric & \multicolumn{12}{c|}{\textbf{$R = 8$}} & \multicolumn{12}{c}{\textbf{$R = 12$}}\\
\cmidrule(lr){3-14}\cmidrule(lr){15-26}
& & \multicolumn{6}{c}{\emph{Soft}} & \multicolumn{6}{c|}{\emph{Consistent}} & \multicolumn{6}{c}{\emph{Soft}} & \multicolumn{6}{c}{\emph{Consistent}} \\
\cmidrule(lr){3-8}\cmidrule(lr){9-14}\cmidrule(lr){15-20}\cmidrule(lr){21-26}
& & \multicolumn{2}{c}{DPS} & \multicolumn{2}{c}{Nila-DC} & \multicolumn{2}{c}{$\Pi$GDM} & \multicolumn{2}{c}{DDNM} & \multicolumn{2}{c}{MCG} & \multicolumn{2}{c|}{Score} & \multicolumn{2}{c}{DPS} & \multicolumn{2}{c}{Nila-DC} & \multicolumn{2}{c}{$\Pi$GDM} & \multicolumn{2}{c}{DDNM} & \multicolumn{2}{c}{MCG} & \multicolumn{2}{c}{Score} \\
\cmidrule(lr){3-4}\cmidrule(lr){5-6}\cmidrule(lr){7-8}\cmidrule(lr){9-10}\cmidrule(lr){11-12}\cmidrule(lr){13-14}\cmidrule(lr){15-16}\cmidrule(lr){17-18}\cmidrule(lr){19-20}\cmidrule(lr){21-22}\cmidrule(lr){23-24}\cmidrule(lr){25-26}
& & base & MSC & base & MSC & base & MSC & base & MSC & base & MSC & base & \multicolumn{1}{c|}{MSC} & base & MSC & base & MSC & base & MSC & base & MSC & base & MSC & base & MSC \\
\midrule
\multirow{3}{*}{T2} & MSD & {\textcolor{black!55}{59.1}} & \textbf{3.9} & {\textcolor{black!55}{24.8}} & \textbf{6.5} & {\textcolor{black!55}{62.4}} & \textbf{9.0} & {\textcolor{black!55}{5.1}} & \textbf{4.0} & {\textcolor{black!55}{3.6}} & \textbf{3.1} & {\textcolor{black!55}{4.1}} & \textbf{3.4} & {\textcolor{black!55}{70.0}} & \textbf{5.8} & {\textcolor{black!55}{29.5}} & \textbf{9.0} & {\textcolor{black!55}{79.9}} & \textbf{16.4} & {\textcolor{black!55}{7.0}} & \textbf{5.6} & {\textcolor{black!55}{5.3}} & \textbf{4.6} & {\textcolor{black!55}{5.9}} & \textbf{5.0} \\
 & PSNR & {\textcolor{black!55}{23.4}} & \textbf{23.5} & {\textcolor{black!55}{22.7}} & \textbf{22.8} & {\textcolor{black!55}{17.9}} & \textbf{19.7} & 22.3 & 22.3 & 22.7 & 22.7 & 22.9 & 22.9 & 20.1 & 20.1 & {\textcolor{black!55}{19.6}} & \textbf{19.7} & {\textcolor{black!55}{16.8}} & \textbf{17.9} & 19.1 & 19.1 & {\textcolor{black!55}{19.6}} & \textbf{19.7} & 19.8 & 19.8 \\
 & SSIM & .76 & .76 & {\textcolor{black!55}{.69}} & \textbf{.72} & {\textcolor{black!55}{.36}} & \textbf{.52} & .70 & .70 & .72 & .72 & {\textcolor{black!55}{.73}} & \textbf{.74} & {\textcolor{black!55}{.63}} & \textbf{.64} & {\textcolor{black!55}{.58}} & \textbf{.60} & {\textcolor{black!55}{.28}} & \textbf{.39} & .57 & .57 & .61 & .61 & .62 & .62 \\
\midrule
\multirow{3}{*}{FLAIR} & MSD & {\textcolor{black!55}{35.5}} & \textbf{4.7} & {\textcolor{black!55}{24.7}} & \textbf{6.8} & {\textcolor{black!55}{56.1}} & \textbf{9.3} & {\textcolor{black!55}{5.8}} & \textbf{4.7} & {\textcolor{black!55}{4.2}} & \textbf{3.8} & {\textcolor{black!55}{4.8}} & \textbf{4.2} & {\textcolor{black!55}{44.5}} & \textbf{6.3} & {\textcolor{black!55}{26.5}} & \textbf{8.5} & {\textcolor{black!55}{65.0}} & \textbf{14.8} & {\textcolor{black!55}{7.3}} & \textbf{6.0} & {\textcolor{black!55}{5.7}} & \textbf{5.0} & {\textcolor{black!55}{6.3}} & \textbf{5.5} \\
 & PSNR & {\textcolor{black!55}{24.6}} & \textbf{24.7} & {\textcolor{black!55}{24.0}} & \textbf{24.1} & {\textcolor{black!55}{19.0}} & \textbf{20.8} & 23.8 & 23.8 & 24.2 & 24.2 & 24.3 & 24.3 & {\textcolor{black!55}{22.2}} & \textbf{22.3} & {\textcolor{black!55}{21.0}} & \textbf{21.1} & {\textcolor{black!55}{17.8}} & \textbf{19.1} & {\textcolor{black!55}{20.5}} & \textbf{20.6} & 21.8 & 21.8 & 22.0 & 22.0 \\
 & SSIM & .71 & .71 & {\textcolor{black!55}{.63}} & \textbf{.68} & {\textcolor{black!55}{.39}} & \textbf{.52} & .68 & .68 & .70 & .70 & .71 & .71 & .62 & .62 & {\textcolor{black!55}{.52}} & \textbf{.56} & {\textcolor{black!55}{.34}} & \textbf{.43} & .54 & .54 & {\textcolor{black!55}{.60}} & .61 & .61 & \textbf{.62} \\
\midrule
\multirow{3}{*}{T1} & MSD & {\textcolor{black!55}{47.6}} & \textbf{2.8} & {\textcolor{black!55}{16.7}} & \textbf{4.1} & {\textcolor{black!55}{53.4}} & \textbf{6.6} & {\textcolor{black!55}{3.5}} & \textbf{2.7} & {\textcolor{black!55}{2.5}} & \textbf{2.1} & {\textcolor{black!55}{2.8}} & \textbf{2.3} & {\textcolor{black!55}{53.0}} & \textbf{3.7} & {\textcolor{black!55}{17.8}} & \textbf{5.2} & {\textcolor{black!55}{60.2}} & \textbf{10.5} & {\textcolor{black!55}{4.3}} & \textbf{3.3} & {\textcolor{black!55}{3.5}} & \textbf{2.9} & {\textcolor{black!55}{3.7}} & \textbf{3.1} \\
 & PSNR & {\textcolor{black!55}{28.6}} & \textbf{28.8} & {\textcolor{black!55}{27.8}} & \textbf{28.0} & {\textcolor{black!55}{20.7}} & \textbf{22.9} & 27.6 & 27.6 & 27.0 & 27.0 & 27.3 & 27.3 & {\textcolor{black!55}{26.2}} & \textbf{26.3} & {\textcolor{black!55}{25.3}} & \textbf{25.4} & {\textcolor{black!55}{19.3}} & \textbf{20.6} & 24.6 & 24.6 & 23.7 & 23.7 & 24.0 & 24.0 \\
 & SSIM & {\textcolor{black!55}{.81}} & \textbf{.82} & {\textcolor{black!55}{.76}} & \textbf{.81} & {\textcolor{black!55}{.56}} & \textbf{.64} & .80 & .80 & .80 & .80 & .80 & .80 & .75 & .75 & {\textcolor{black!55}{.71}} & \textbf{.74} & {\textcolor{black!55}{.49}} & \textbf{.55} & .71 & .71 & .70 & .70 & .71 & .71 \\
\midrule
\multirow{3}{*}{T1PRE} & MSD & {\textcolor{black!55}{61.3}} & \textbf{4.0} & {\textcolor{black!55}{23.9}} & \textbf{6.1} & {\textcolor{black!55}{78.2}} & \textbf{10.9} & {\textcolor{black!55}{5.2}} & \textbf{4.1} & {\textcolor{black!55}{3.8}} & \textbf{3.2} & {\textcolor{black!55}{4.2}} & \textbf{3.6} & {\textcolor{black!55}{66.2}} & \textbf{5.1} & {\textcolor{black!55}{24.8}} & \textbf{7.6} & {\textcolor{black!55}{85.0}} & \textbf{16.6} & {\textcolor{black!55}{6.4}} & \textbf{5.2} & {\textcolor{black!55}{4.6}} & \textbf{4.1} & {\textcolor{black!55}{5.2}} & \textbf{4.5} \\
 & PSNR & {\textcolor{black!55}{26.8}} & \textbf{27.0} & {\textcolor{black!55}{26.4}} & \textbf{26.6} & {\textcolor{black!55}{18.8}} & \textbf{21.2} & 26.2 & 26.2 & 26.4 & 26.4 & {\textcolor{black!55}{26.5}} & \textbf{26.6} & 24.1 & 24.1 & {\textcolor{black!55}{23.4}} & \textbf{23.5} & {\textcolor{black!55}{17.2}} & \textbf{18.7} & {\textcolor{black!55}{22.9}} & \textbf{23.0} & 22.8 & 22.8 & 23.1 & 23.1 \\
 & SSIM & {\textcolor{black!55}{.77}} & \textbf{.78} & {\textcolor{black!55}{.73}} & \textbf{.77} & {\textcolor{black!55}{.48}} & \textbf{.56} & .76 & .76 & {\textcolor{black!55}{.76}} & \textbf{.77} & .77 & .77 & .69 & .69 & {\textcolor{black!55}{.64}} & \textbf{.68} & {\textcolor{black!55}{.42}} & \textbf{.47} & .66 & .66 & .66 & .66 & .67 & .67 \\
\midrule
\multirow{3}{*}{T1POST} & MSD & {\textcolor{black!55}{122.5}} & \textbf{4.2} & {\textcolor{black!55}{34.0}} & \textbf{5.2} & {\textcolor{black!55}{115.0}} & \textbf{9.5} & {\textcolor{black!55}{4.8}} & \textbf{4.1} & {\textcolor{black!55}{3.8}} & \textbf{3.4} & {\textcolor{black!55}{4.2}} & \textbf{3.7} & {\textcolor{black!55}{122.9}} & \textbf{5.6} & {\textcolor{black!55}{34.7}} & \textbf{6.3} & {\textcolor{black!55}{120.4}} & \textbf{14.9} & {\textcolor{black!55}{5.7}} & \textbf{4.8} & {\textcolor{black!55}{5.3}} & \textbf{4.8} & {\textcolor{black!55}{5.7}} & \textbf{5.0} \\
 & PSNR & {\textcolor{black!55}{27.9}} & \textbf{28.1} & {\textcolor{black!55}{26.7}} & \textbf{26.9} & {\textcolor{black!55}{19.9}} & \textbf{22.2} & {\textcolor{black!55}{26.6}} & \textbf{26.7} & 26.7 & 26.7 & {\textcolor{black!55}{26.8}} & \textbf{26.9} & {\textcolor{black!55}{24.8}} & \textbf{24.9} & {\textcolor{black!55}{24.2}} & \textbf{24.3} & {\textcolor{black!55}{18.2}} & \textbf{19.7} & 23.7 & 23.7 & 22.9 & 22.9 & 23.3 & 23.3 \\
 & SSIM & {\textcolor{black!55}{.79}} & \textbf{.80} & {\textcolor{black!55}{.73}} & \textbf{.78} & {\textcolor{black!55}{.52}} & \textbf{.60} & .77 & .77 & .78 & .78 & .78 & .78 & .71 & .71 & {\textcolor{black!55}{.66}} & \textbf{.70} & {\textcolor{black!55}{.45}} & \textbf{.50} & .68 & .68 & .65 & .65 & .67 & .67 \\
\bottomrule
\end{tabular}
\end{table*}

\begin{figure*}[t]
\centering
\includegraphics[width=0.98\linewidth]{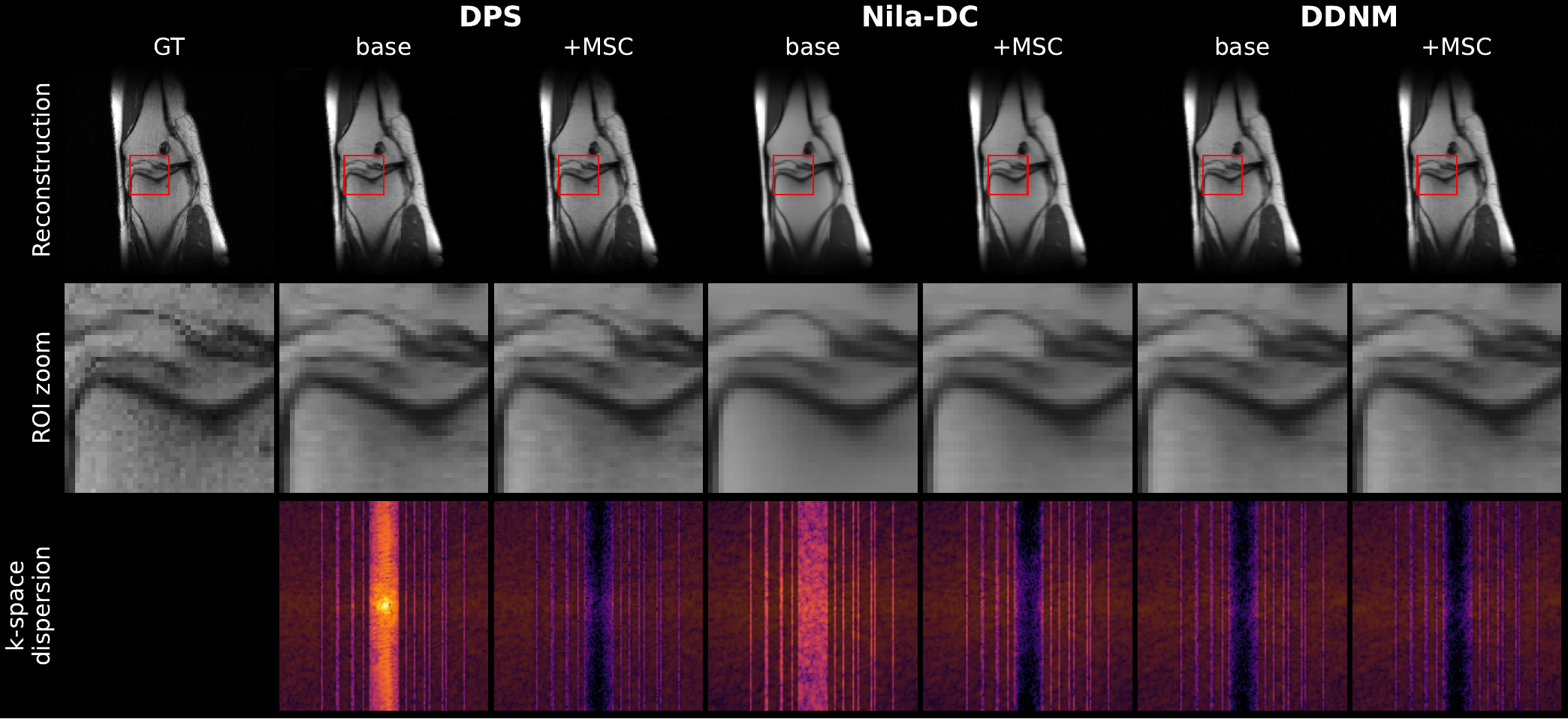}
\caption{\textbf{Qualitative measured-subspace consistency.} In-domain knee PD, Gaussian-1D mask, $R=8$. Columns show GT and base/+MSC pairs for DPS, Nila-DC, and DDNM. Rows show chain-mean reconstruction, ROI zoom, and central k-space posterior dispersion across $L$ samples with unmeasured lines dimmed.}
\label{fig:qualitative}
\end{figure*}

\section{Experiments}
\label{sec:exp}
\subsection{Experimental Setup \& Baselines}
We evaluate on two complementary multi-coil fastMRI~\cite{zbontar2018fastmri} anatomies: axial brain (a high-SNR clinical dataset) and coronal knee, which serves as an anatomy-transfer test. All data are coil-combined with ESPIRiT-estimated sensitivities~\cite{uecker2014espirit}, square-cropped and resized to $256{\times}256$, and scored on a foreground ROI after a single least-squares scale match. The full brain sweep (\cref{tab:brain}) spans five contrasts. The knee sweep covers the PD and PDFS contrasts. Each cell is the median over 20 volumes $\times$ 4 central slices, and each experiment states the contrasts it uses (e.g.\ the downstream task reports T1 and T1POST). Unless otherwise stated, headline analyses use acceleration $R = 8$, $L = 8$ chains, and $T = 300$ reverse steps.

We wrap six base samplers spanning the Soft-to-Consistent spectrum. The three \emph{Soft} bases steer sampling with a soft data term and lack a hard measured-subspace projection: \textbf{DPS}~\cite{chung2023dps}, \textbf{Nila-DC}~\cite{nila2024} (noise-adaptive data consistency), and \textbf{$\Pi$GDM}~\cite{song2023pseudoinverse}. The three bases on the \emph{Consistent} side either embed a hard projection or empirically behave measured-consistently in our audit: \textbf{DDNM}~\cite{wang2023ddnm} (range/null-space split), \textbf{MCG}~\cite{chung2022mcg} (manifold-constrained gradient), and \textbf{Score-MRI}~\cite{chung2022scoremri,song2021solving} (per-step POCS-style k-space replacement). Each base appears beside its deterministic $+$MSC wrap. MSC is not a new reconstructor: reconstruction quality is inherited from the base sampler, while the targeted effect is measured-subspace admissibility. The \textbf{DC-only} variant, a measured lock without stochasticity, serves as a diversity-collapse control.

No base sampler is retuned when MSC is appended. Each base uses its source paper's default guidance and step settings, and the same $L=8$ noise seeds are shared between a base and its $+$MSC wrap, so each comparison is paired and MSC is a deterministic post-map of the final chains. Sensitivity maps use sigpy ESPIRiT rather than the BART reference, so absolute numbers are not directly comparable to the original Nila report. Both priors are unconditional DDPM-style U-Nets~\cite{ho2020ddpm} ($256{\times}256$, $128$ base channels, $1000$-step linear schedule, EMA weights), trained separately per anatomy following the Nila setup~\cite{nila2024}: a multi-contrast brain prior on 52k fastMRI images and a single-contrast PD knee prior on 30.8k fastMRI images. Since each prior is anatomy-specific, applying one to the other is a genuine out-of-distribution test on an anatomy the prior never saw. Headline numbers carry volume-level bootstrap 95\% confidence intervals, and paired significance tests use volume-level Wilcoxon tests with Holm correction where multiple leakage-prone pairs are tested.

\begin{table}[t]
\centering\small
\caption{\textbf{Reversed organ mismatch.} Brain prior reconstructs knee CORPD at $R=8$. Metrics: MSD ($\times10^{-4}$, $\downarrow$), PSNR (dB, $\uparrow$), SSIM ($\uparrow$, leading zero omitted). Better of base/+MSC in \textbf{bold} (display ties unbolded).}
\label{tab:oodknee}
\setlength{\tabcolsep}{4pt}
\begin{tabular}{l|cc|cc|cc}
\toprule
& \multicolumn{2}{c|}{MSD$\downarrow$} & \multicolumn{2}{c|}{PSNR$\uparrow$} & \multicolumn{2}{c}{SSIM$\uparrow$}\\
\cmidrule(lr){2-3}\cmidrule(lr){4-5}\cmidrule(lr){6-7}
& base & MSC & base & MSC & base & MSC\\
\midrule
\multicolumn{7}{l}{\emph{Soft}}\\
DPS & {\textcolor{black!55}{40.2}} & \textbf{5.1} & 28.3 & 28.3 & {\textcolor{black!55}{.78}} & \textbf{.79}\\
Nila-DC & {\textcolor{black!55}{24.3}} & \textbf{7.0} & {\textcolor{black!55}{27.4}} & \textbf{27.5} & {\textcolor{black!55}{.73}} & \textbf{.76}\\
$\Pi$GDM & {\textcolor{black!55}{51.7}} & \textbf{9.3} & {\textcolor{black!55}{22.0}} & \textbf{23.3} & {\textcolor{black!55}{.56}} & \textbf{.63}\\
\midrule
\multicolumn{7}{l}{\emph{Consistent}}\\
DDNM & {\textcolor{black!55}{6.3}} & \textbf{5.0} & 26.8 & 26.8 & .74 & .74\\
MCG & {\textcolor{black!55}{4.5}} & \textbf{4.0} & {\textcolor{black!55}{27.0}} & \textbf{27.1} & .76 & .76\\
Score & {\textcolor{black!55}{5.1}} & \textbf{4.4} & 27.2 & 27.2 & .76 & .76\\
\bottomrule
\end{tabular}
\end{table}

\subsection{Main Reconstruction Results: Soft vs.\ Consistent Split}
\label{sec:structural}
\Cref{tab:knee} (in-domain knee, across masks, accelerations, and the PD/PDFS contrasts) and \cref{tab:brain} (out-of-distribution transfer, where the knee prior reconstructs five brain contrasts) report MSD$\downarrow$ and reconstruction quality (PSNR/SSIM) for every base sampler, comparing each base against its paired $+$MSC column. MSD is the primary admissibility metric, with PSNR/SSIM as quality checks. The main result is a sampler-structure effect: MSC changes samplers that leak in the measured subspace and is nearly an identity map on samplers that already enforce a hard measured-subspace constraint.

\noindent\textbf{Finding 1: MSC selectively reduces MSD on Soft samplers.} Wrapping the leakage-prone Soft bases (DPS, Nila-DC, $\Pi$GDM) cuts MSD by about an order of magnitude across datasets and conditions. In contrast, the already-consistent bases (DDNM, MCG, Score-MRI) move only marginally, matching the fixed point of \cref{prop:exact} rather than a generic re-projection artifact. Volume-level paired tests confirm every leakage-prone MSD drop ($p \le 1.9 \times 10^{-6}$ after Holm correction), and bootstrap confidence intervals for the headline DPS-brain reduction are tight (e.g.\ T2 $15.0\times$ $[14.1,16.1]$, median $16.5\times$ across the five contrasts, up to ${\sim}29\times$ on T1POST). These ratios and their intervals are computed per volume and then aggregated, so they need not equal the quotient of the independently rounded median cells in \cref{tab:brain}. \Cref{fig:qualitative} shows this effect at the sample level: on the Soft bases the dispersion on acquired k-space lines drops sharply under $+$MSC while the reconstruction is visually unchanged, whereas the Consistent DDNM barely moves.

\noindent\textbf{Finding 2: The MSD reduction is measured-subspace consistency, not posterior collapse.} A measured-subspace projection can drive MSD near zero trivially by collapsing the posterior, so MSD must be read together with unmeasured-subspace dispersion (USD). MSC preserves the base sampler's unmeasured diversity while suppressing measured leakage, whereas a hard data-consistency lock without stochasticity collapses it. The ablation (\cref{tab:placement}) quantifies this gap.

\noindent\textbf{Finding 3: Reconstruction quality is preserved or modestly improved.} Across both datasets, MSC holds or improves PSNR and SSIM. The strongest bases change little in PSNR (typically within $\pm0.1$\,dB), while the weakest base, $\Pi$GDM, gains the most in both PSNR and SSIM. Paired tests confirm SSIM gains on every leakage-prone base, including DPS on knee PD and brain T2 ($20/20$ volumes in both, $p = 1.9\times10^{-6}$).

\begin{table}[t]
\centering\small
\caption{\textbf{Ablation on the DPS base.} In-domain knee PD and OOD brain T1POST at $R=8$. Columns toggle sampler, output lock, and SENSE pseudoinverse ($\mathbf{S}^+$); DC-only is the diversity-collapse control. Bold compares base DPS with MSC rows only (display ties unbolded).}
\label{tab:placement}
\renewcommand{\arraystretch}{0.85}
\setlength{\tabcolsep}{3pt}
\begin{tabular}{ccc|cccc}
\toprule
sampler & out & $\mathbf{S}^{+}$ & MSD$\downarrow$ & USD\,(\%)$\uparrow$ & PSNR$\uparrow$ & SSIM$\uparrow$\\
\midrule
\multicolumn{7}{l}{\emph{knee PD} (in-domain)}\\
\cmark &        &        & {\textcolor{black!55}{30.87}} & {\textcolor{black!55}{100\%}}  & {\textcolor{black!55}{30.91}} & {\textcolor{black!55}{0.849}}\\
       & \cmark &        & {\textcolor{black!55}{0.40}}  & {\textcolor{black!55}{7.5\%}}  & {\textcolor{black!55}{30.30}} & {\textcolor{black!55}{0.846}}\\
\cmark & \cmark &        & \textbf{4.73} & 99.3\% & 30.91 & \textbf{0.852}\\
\cmark & \cmark & \cmark & \textbf{4.68} & 98.3\% & 30.91 & \textbf{0.852}\\
\midrule
\multicolumn{7}{l}{\emph{brain T1POST} (OOD)}\\
\cmark &        &        & {\textcolor{black!55}{122.5}} & {\textcolor{black!55}{100\%}} & {\textcolor{black!55}{27.91}} & {\textcolor{black!55}{0.795}}\\
       & \cmark &        & {\textcolor{black!55}{0.2}} & {\textcolor{black!55}{7.6\%}} & {\textcolor{black!55}{27.92}} & {\textcolor{black!55}{0.799}}\\
\cmark & \cmark &        & \textbf{4.2} & 99.7\% & \textbf{28.08} & \textbf{0.798}\\
\cmark & \cmark & \cmark & \textbf{4.2} & 98.7\% & \textbf{28.08} & \textbf{0.798}\\
\bottomrule
\end{tabular}
\end{table}

\subsection{Robustness \& Generalization Under Mismatch}
\label{sec:general}
The robustness experiments vary the axes that could plausibly break a measured-subspace post-map. First, mask type does not require retuning: \cref{tab:knee} includes Uniform 1D, Gaussian-1D, Gaussian-2D, and Poisson-disc masks, and MSC simply locks the samples to whichever measurements were acquired. Second, the pattern persists across acceleration, covering knee PD/PDFS at $R\in\{4,8,12,15\}$ and brain at $R\in\{8,12\}$.

The same organization holds under anatomy and prior mismatch. \Cref{tab:knee} covers knee reconstruction, while \cref{tab:brain} reconstructs brain contrasts using the knee prior. The tables therefore test both in-domain use and out-of-distribution transfer. The trend persists even under a reversed organ mismatch: applying the brain prior to knee data still cuts measured-subspace dispersion on the Soft bases by $3.5$--$7.9\times$ in \cref{tab:oodknee}, while holding or improving SSIM, and leaves the Consistent bases near-unchanged. Finally, the held-out samplers land on their predicted sides: $\Pi$GDM behaves as a Soft sampler, while MCG behaves as a Consistent sampler.

\begin{table}[t]
\centering
\caption{\textbf{OOD downstream brain extraction.} HD-BET Dice ($\uparrow$) and HD95 (mm, $\downarrow$) for DPS base vs.\ +MSC on T1/T1POST at $R=8,12$. The vol. row counts higher-Dice MSC volumes. Better values in \textbf{bold} (display ties unbolded).}
\label{tab:downstream}
\small
\renewcommand{\arraystretch}{0.85}
\setlength{\tabcolsep}{3pt}
\begin{tabular}{l|l|cc|cc}
\toprule
\textbf{Contrast} & Metric & \multicolumn{2}{c|}{$R = 8$} & \multicolumn{2}{c}{$R = 12$} \\
\cmidrule(lr){3-4}\cmidrule(lr){5-6}
& & base & MSC & base & MSC \\
\midrule
\multirow{3}{*}{T1} & Dice$\uparrow$ & {\textcolor{black!55}{.950}} & \textbf{.951} & {\textcolor{black!55}{.933}} & \textbf{.934} \\
 & HD95$\downarrow$ & {\textcolor{black!55}{4.82}} & \textbf{4.76} & {\textcolor{black!55}{6.02}} & \textbf{5.90} \\
 & vol.$\uparrow$ & \multicolumn{2}{c|}{\textbf{14/20}} & \multicolumn{2}{c}{\textbf{14/20}} \\
\cmidrule(lr){1-6}
\multirow{3}{*}{T1POST} & Dice$\uparrow$ & .956 & .956 & {\textcolor{black!55}{.928}} & \textbf{.929} \\
 & HD95$\downarrow$ & {\textcolor{black!55}{4.31}} & \textbf{4.20} & {\textcolor{black!55}{6.42}} & \textbf{6.29} \\
 & vol.$\uparrow$ & \multicolumn{2}{c|}{\textbf{15/20}} & \multicolumn{2}{c}{\textbf{12/20}} \\
\bottomrule
\end{tabular}
\end{table}

\subsection{Ablation Study}
\label{sec:placement}
The ablation separates MSC from two simpler alternatives. On the DPS base (\cref{tab:placement}), MSC lowers MSD in both knee PD and out-of-distribution brain T1POST while keeping nearly all unmeasured diversity ($\geq98\%$). A DC-only measured lock reaches near-zero MSD only by collapsing that diversity ($\leq8\%$), so it is not an admissible uncertainty-preserving substitute. Replacing the adjoint coil combiner $\mathbf{S}^{H}$ with the SENSE pseudoinverse $\mathbf{S}^{+}$ leaves the result essentially unchanged, as expected since RSS-normalized maps give $\mathbf{S}^{+}=\mathbf{S}^{H}$. The small residual reflects coil combination ($\mathbf{S}\mathbf{S}^{H}\neq\mathbf{I}$), not the combiner choice, so the default keeps the simpler $\mathbf{S}^{H}$.

\subsection{Downstream Task Utility \& Computational Overhead}
\label{sec:downstream}
Brain extraction provides a downstream sanity check: removing measured-subspace leakage should not damage task utility. Using the same shared PD knee prior as \cref{tab:knee,tab:brain}, we take the fully sampled brain mask as reference and measure the HD-BET~\cite{isensee2019hdbet} mask on 20 full-volume fastMRI brain cases for T1 and T1POST (\cref{tab:downstream}). The DPS base is already near the mask ceiling (Dice $0.93$--$0.96$), and MSC is neutral: Dice changes by at most $\pm0.001$, while HD95 is unchanged or slightly improved (e.g.\ T1 $6.02\to5.90$, T1POST $6.42\to6.29$\,mm at $R = 12$).

MSC adds one output correction: a forward/inverse FFT and coil combination ($O(C\,N\log N)$ for $C$ coils and $N$ pixels), with zero trainable parameters. On a $256{\times}256$ A/B harness with 8 chains, one application takes 1.09\,ms versus 73.0\,ms for one denoiser reverse step. Applied once after a $T = 300$ trajectory, this is $<0.01\%$ overhead. It requires no retraining and no base-sampler retuning.

%% file: sec/5_conclusion.tex
\section{Conclusion}
\label{sec:conclusion}
We introduced measured-subspace consistency (MSC), a training-free, plug-and-play terminal correction that locks acquired multi-coil k-space at a diffusion posterior sampler's output. Theoretically, we prove exact measured-subspace agreement in the ideal case and bound the cross-subspace coupling residual for the practical sensitivity-weighted operator. Empirically, we demonstrate MSC's selectivity across six base samplers, two anatomies, and out-of-distribution priors: it substantially reduces measured-subspace dispersion for leakage-prone Soft bases while acting as a near-identity map for Consistent ones. By suppressing measured-subspace leakage, MSC preserves legitimate unmeasured diversity and maintains or improves reconstruction quality at negligible cost. It thus redirects posterior variability away from scanner-fixed coefficients and toward unacquired frequencies, ensuring that the sampler's uncertainty reflects what the prior had to infer.

\paragraph{Limitations and Future Work.} First, in the presence of measurement noise ($\y=\mathbf{A}\x+\boldsymbol{\varepsilon}$), locking the acquired k-space introduces a small conservative bias by removing legitimate noise-level variance in the measured subspace. Second, the sensitivity-weighted multi-coil projection leaves a small residual measured-subspace floor because $\mathbf{S}\mathbf{S}^H\neq\mathbf{I}$ couples a fraction of unmeasured energy onto measured lines. Third, our analysis is restricted to 2D Cartesian multi-coil MRI. Extending MSC to non-Cartesian trajectories, 3D acquisitions, and evaluating its prospective clinical utility in downstream diagnostic tasks are important directions for future work.